\documentclass[11pt]{article}

\usepackage[margin=1in]{geometry}
\usepackage[T1]{fontenc}
\usepackage{amsmath,amssymb,amsthm,mathtools}
\usepackage{braket}
\usepackage{graphicx}
\usepackage{xcolor}
\usepackage{tikz}
\usepackage{booktabs}
\usepackage{enumerate}
\usepackage{microtype}
\usepackage[colorlinks=true,linkcolor=blue,citecolor=blue,urlcolor=blue]{hyperref}
\usepackage[capitalise,nameinlink]{cleveref}

\usetikzlibrary{calc}

\newtheorem{theorem}{Theorem}[section]
\newtheorem{lemma}[theorem]{Lemma}
\newtheorem{proposition}[theorem]{Proposition}

\theoremstyle{definition}
\newtheorem{definition}[theorem]{Definition}
\newtheorem{conjecture}[theorem]{Conjecture}
\theoremstyle{remark}

\newcommand{\E}{\mathbb{E}}

\newcommand{\ALG}{\operatorname{ALG}}
\newcommand{\OPT}{\operatorname{OPT}}
\newcommand{\cL}{\mathcal{L}}
\newcommand{\cP}{\mathcal{P}}
\newcommand{\cI}{\mathcal{I}}

\newcommand{\R}{\mathbb{R}}
\newcommand{\C}{\mathbb{C}}
\newcommand{\N}{\mathbb{N}}

\newcommand{\keywords}[1]{}

\title{Online geometric matching and segments block points}
\author{Yaqiao Li\footnote{liyaqiao@suat-sz.edu.cn}  \ \ 
Junyi Lin\footnote{suat25000108@stu.suat-sz.edu.cn}  \\
Shenzhen University of Advanced Technology,  Shenzhen, China}

\date{\today}

\begin{document}

\maketitle

\begin{abstract}
Given a sequence of online points in general position in the plane, the geometric matching problem asks to match as many points as possible, subject to segment crossing constraints. We initiate the study of online geometric matching with limited crosses. In the $k$-crossing model each matched segment is allowed to have at most $k$ crosses, where $k=0$ corresponds to the online noncrossing matching model for which Bose et al. proved that Greedy is an optimal deterministic algorithm achieving a competitive ratio of $2/3$. In the triangle-free-crossing model no three matched segments are allowed to pairwise cross each other. Interestingly, we show in both cases the performance of Greedy does not improve by allowing crosses, and in fact, becomes worse in the $k$-crossing setting for $k\ge 1$. In the triangle-free-crossing model, we propose a simple two-layer-Greedy deterministic algorithm achieving a competitive ratio $4/5$. More interestingly, we propose and study the underlying segments blocking points extremal problem, and discuss its connection to vertex cover numbers and Ramsey numbers, with interesting open problems that might be of independent interest.

We revisit the noncrossing model. We give a randomized algorithm which is an asymmetric version of Kamali et al., improving the competitive ratio from $0.66952$ to $0.67034$. We also show, via a conceptually simple construction, that no randomized algorithm can match more than $6/7$ fraction of points, improving the previous $16/17$ by Boyar et al.

\end{abstract}

\section{Introduction}\label{sec:introduction}

This work generalizes the online noncrossing matching problem to the more general online geometric matching problem, where limited crosses are allowed. The online \emph{geometric} matching problem differs from online  matching of a \emph{graph}:
for the geometric problem, the matching needs to satisfy a certain geometric property, while for the graph problem, the matching needs to satisfy a combinatorial property (i.e., be a subset of predetermined edges). For online matching of a graph, see a recent survey by Huang et al.~\cite{huang2024online}. 

The online geometric matching problem is defined as follows. An online instance is a finite sequence
$\cI=(v_1,\ldots,v_n)$ of points in $\R^2$ in general position; in particular no three input points are collinear, no three lines intersect at one point, etc.  The value of $n$ is not announced in
advance.  When $v_i$ arrives, an online algorithm may match it to one earlier
unmatched point or leave it unmatched.  A match is represented by the closed
straight-line segment between its endpoints, and all decisions are
irrevocable.  Throughout, for two points $u$ and $v$, we write $uv$ for the closed line segment joining them.  The matched segments form a matching denoted by $\cL$. For a given geometric property, one asks: what is the largest asymptotic competitive ratio that can be achieved by an online algorithm respecting the given property? 

In this work, the property we consider is whether two line segments cross or not, i.e., whether their  interiors
intersect. 
The \emph{crossing graph} $G(\cL)$ has one vertex for every segment of $\cL$
and one edge for every crossing pair.  We consider two scenarios.
\begin{itemize}
\item In the \emph{$k$-crossing} model, every vertex of $G(\cL)$ has degree at
      most the fixed integer $k\geq0$.
\item In the \emph{triangle-free-crossing} model, $G(\cL)$ is triangle-free;
      equivalently, no three matching segments pairwise cross.
\end{itemize}
The $0$-crossing model is precisely the \emph{noncrossing} model.  For $k=0,1$, every
$k$-crossing matching is triangle-free-crossing.  For $k\geq2$, the two
crossing models are incomparable: a triangle-free crossing graph may have
arbitrarily large maximum degree, whereas a $k$-crossing graph may contain a
triangle.

Let $\ALG$ denote an arbitrary online algorithm, and $\OPT$ an optimal offline algorithm, respectively. For a given instance $\cI$, let $\ALG(\cI)$ and $\OPT(\cI)$ denote the number of matched points, respectively. We are only concerned about the asymptotic performance of algorithms, so without loss of generality we assume $n$ is even. Because a minimum-total-length perfect matching is noncrossing, we have $\OPT(\cI) = n$ for all three models we consider. If $\cL$ is the matching obtained by $\ALG$, then 
$\ALG(\cI)=2|\cL|$, where $|\cL|$ denotes the number of matched segments.
A deterministic algorithm has asymptotic competitive ratio at least $\rho$ if
a constant $C$ exists such that
$\ALG(\cI)\geq\rho\OPT(\cI)-C$ for every input $\cI$.  For a randomized
algorithm the corresponding requirement is
$\E[\ALG(\cI)]\geq\rho\OPT(\cI)-C$, where the expectation is over the
algorithm's internal randomness.  Its adversary is oblivious: the entire
input, including its length, is fixed before the random bits are exposed.
Since $\OPT(\cI)=n$ in all three models, the asymptotic competitive ratio coincides with the asymptotic fraction of input points matched (or the expected fraction for a randomized algorithm).

A deterministic algorithm is \emph{Greedy} if it always matches the current
point whenever that point can be matched to an existing unmatched point.  Its
tie-breaking rule is arbitrary.  

\subsection{Our contribution}   \label{sec:contri}

We defer a more detailed discussion of related work and motivation to the next subsection, but only mention the following that is most directly related to this work: Bose et al.~\cite{bose2020non} initiated the online noncrossing matching problem and
proved that Greedy has competitive ratio at least $2/3$, which is optimal among deterministic algorithms.  Kamali, Nikbakht, and Sajadpour
~\cite{Kamali2022randomizedNM} gave the first randomized algorithm with expected competitive ratio
$235/351\approx0.66952$.  Boyar et al.~\cite{boyar2024online} proved that no randomized algorithm can have competitive ratio greater than $16/17$.  Our main contributions are:

\begin{itemize}
    \item We initiate the study of the $k$-crossing and the triangle-free-crossing matching. Intuitively, by allowing crosses it seems Greedy may match more points. We show this is false: for $k$-crossing matching Greedy has a competitive ratio of at most $4/(6+k)$, which becomes strictly worse compared to the noncrossing model for every $k \ge 1$. For triangle-free-crossing matching Greedy has a competitive ratio of at most $2/3$, and we show this by presenting an elegant construction. Surprisingly, proving a matching lower bound seems difficult and is left as the main open problem of this work. Furthermore, for triangle-free-crossing matching, we introduce a simple deterministic two-layer-Greedy online algorithm with competitive ratio $4/5$ which strictly improves Greedy's performance in the noncrossing setting.
    
    Perhaps more interestingly, the performance of Greedy in both settings can be understood through a common geometric extremal problem. In particular, for the triangle-free-crossing matching, our study of the corresponding geometric extremal problem indicates a potential connection to Ramsey theory.  

    \item We give a refined randomized algorithm based on ~\cite{Kamali2022randomizedNM} with an expected competitive ratio $0.67034$, improving their $0.66952$. We also show, via a conceptually simpler construction than that of Boyar et al.~\cite{boyar2024online}, that no randomized algorithm can have competitive ratio greater than $6/7 \approx 0.85714$, improving their $16/17 \approx 0.94118$.  
\end{itemize}

We proceed to define the mathematical extremal problem in its general form, which we believe may also be of independent interest.

\begin{definition}\label{def:blocking}
Let $\cL$ be a family of line segments having some geometric property $\mathsf P$. Let $\cP$ be a set of points such that its points and the endpoints
of $\cL$ are jointly in general position.  The family $\cL$ \emph{blocks} $\cP$, if $\cL\cup \{uv\}$ violates $\mathsf P$ for every two distinct points
$u,v\in\cP$.
Write
\[
                         \ell=|\cL|,\qquad p=|\cP|.
\]
The blocking number $\beta_{\mathsf P}(\cL)$ is the largest cardinality of a
point set blocked by $\cL$, and define
\[
 \beta_{\mathsf P}(\ell) :=
 \max\{\beta_{\mathsf P}(\cL):|\cL|=\ell,
                 \ \cL\text{ has property $\mathsf P$}\}.
\]
\end{definition}

In this work, $\mathsf P$ is either the $k$-crossing, or the
triangle-free-crossing property.  We write $\beta_k$ for the $k$-crossing
property, including $k=0$, and $\beta_\triangle$ for triangle-free crossing. See Figure \ref{fig:triangle-blocking} for an illustration, which shows $\beta_\triangle(5) \ge 4$.

\begin{figure}[ht!]
\centering
\includegraphics[width=.4\linewidth]{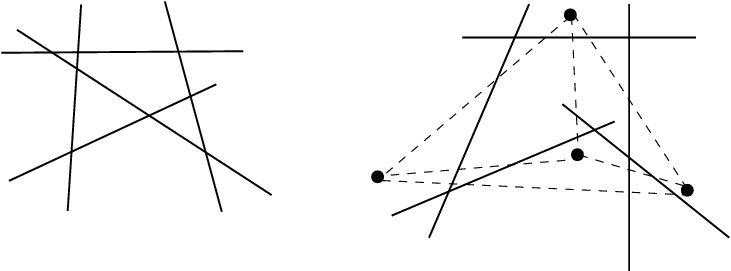}
\caption{Left: a configuration of segments that is \emph{not} triangle-free. Right: a triangle-free-crossing family of five segments blocks four points, hence $\beta_\triangle(5) \ge 4$.}
\label{fig:triangle-blocking}
\end{figure}

The connection between segments blocking points and Greedy algorithms is exhibited in the following. We remark that this connection can be stated for any other geometric property $\mathsf P$, but we need to determine the corresponding $\OPT$ for the property $\mathsf P$. We only state this connection for the properties that concern us in this work, because $\OPT$ in our case is particularly simple.

\begin{theorem} \label{thm:block-to-greedy}
    Let $\mathsf P$ be $k$-crossing or triangle-free-crossing. 
    \begin{enumerate}[(1)]
        \item If $\beta_{\mathsf P}(\ell) \le p$, then on every input on which Greedy creates $\ell$ segments, Greedy matches at least $2\ell /(2\ell + p)$ fraction of the input points.

        \item If $\beta_{\mathsf P}(\ell) \ge p$, then there exists an input on which Greedy creates $\ell$ segments and matches at most $2\ell /(2\ell + p)$ fraction of the input points.
    \end{enumerate}
\end{theorem}

\begin{proof}
    (1) At the end of an online instance on which Greedy creates $\ell$ segments, let $\cL$ be the family of segments matched by Greedy and let $\cP$ be the set of unmatched points. Because of the greedy nature of Greedy, $\cL$ blocks $\cP$. Hence,
        $|\cP| \le \beta_{\mathsf P}(\ell) \le p$. 
    The total number of points is at most $2\ell + p$, hence the claimed bound.

    (2) In this case, there exists a family $\cL$ of $\ell$ segments that blocks a set containing at least $p$ points; choose $p$ of those points and call the resulting set $\cP$. Present the $2\ell$ endpoints of $\cL$ pair by pair, with the endpoints of each prescribed segment presented consecutively, followed by the $p$ points of $\cP$. Because $\cL$ satisfies property $\mathsf P$, Greedy matches all $2\ell$ endpoints into the $\ell$ prescribed segments. Since $\cL$ blocks $\cP$, none of the points in $\cP$ can be matched by Greedy. The claim follows.
\end{proof}

Because of Theorem \ref{thm:block-to-greedy}, to understand the performance of Greedy it suffices to study the following extremal problem.

\smallskip

{\noindent\bf Problem:} determine $\beta_{\mathsf P}(\ell)$, and understand the extremal configuration $\cL$.

\smallskip

The rest of the paper is organized as follows. Section~\ref{sec:k-crossing} studies Greedy in the $k$-crossing model, where we determine its exact competitive ratio $4/(k+6)$. Section~\ref{sec:triangle-free} studies Greedy in the triangle-free-crossing model, including the two-layer-Greedy algorithm with ratio $4/5$ and the connection to the vertex cover number and to Ramsey numbers. Section~\ref{sec:noncrossing} revisits the noncrossing model, that is, the $0$-crossing case, where we improve the randomized lower bound of Kamali et al.~\cite{Kamali2022randomizedNM} and prove the $6/7$ upper bound. Section~\ref{sec:open} lists several open problems.

\subsection{Related work}\label{sec:related}

Given $2n$ points in $\R^2$ in general position, the basic noncrossing matching (NM) problem is to find a noncrossing matching with the largest possible number of segments. Observe that the minimum-length Euclidean matching is noncrossing. Hence, a perfect NM always exists. The NM problem and its variants have been extensively studied in the offline setting. 

A closely related line of research relaxes the noncrossing constraint and allows certain crossings. An important problem is to understand the size of a crossing family, that is, matching segments that are pairwise crossing. A breakthrough by Pach et al.~\cite{Pach2019crossingfamily} showed that the largest crossing family has almost linear size. Aichholzer et al.~\cite{Aichholzer2023kcrossing} studied the counting problem of perfect matchings when at most $k$ crossing points are allowed. Note that this differs from our $k$-crossing matching: their $k$ denotes the global number of crossings. In contrast, the online models studied in this paper control the crossings per segment: in the $k$-crossing model every matched segment may cross at most $k$ other segments (Section~\ref{sec:k-crossing}), while in the triangle-free-crossing model no three matched segments pairwise cross (Section~\ref{sec:triangle-free}).

The interest in geometric settings, particularly the Euclidean plane setting, for the matching problem stems from applications in image processing Cohen~\cite{cohen1999finding} and circuit board design Hershberger and  Suri~\cite{hershberger1997efficient}. In such applications, one is often required to construct a matching between various geometric shapes, such as rectangles or circles, representing vertices, using straight-line segments or, more generally, curves. Geometry enters the picture due to constraints on the segments, such as avoiding intersections among the segments.

Hershberger and Suri~\cite{HS1992Alg} gave an algorithm that finds a perfect NM in time $\Theta(n\log n)$. First Atallah~\cite{atallah1985matching} and then Dumitrescu and Steiger~\cite{dumitrescu2000matching} gave efficient algorithms for the bichromatic version of the problem, where the points are divided into two subsets, and matching segments can only be formed between the two subsets. Other versions of the problem considered in the literature include requiring the NM to be stable~\cite{Ruangwises2019StableNM,Hamada2021StableNM2}, requiring two NMs to be compatible (segments in two NMs are also noncrossing if they only share endpoints)~\cite{Aichholzer2009compatibleNM}, and requiring compatible NMs to satisfy a certain diversity constraint~\cite{Misra2022diverseNM}. Noncrossing matching consisting of hyperedges in hypergraphs has also been studied~\cite{kiwi2011generalizations}. 

Several studies considered optimization problems over all NMs. The objective functions include maximizing the sum of the Euclidean lengths of the matching segments~\cite{Alon1993longNM}, minimizing the length of the longest matching segment~\cite{Abu2014bottleneck}, and other similar combinations of minimization and maximization~\cite{Mantas2024Variants}.

At least two works~\cite{Balogh2005WNM,SU2011heavyNM} considered weighted NM on $n$ points, where every point has weight in $\{1,2,\ldots, n\}$. Balogh et al.~\cite{Balogh2005WNM} considered the weight of a segment to be the sum of the weights of the two endpoints modulo $n$, and studied the typical size of an NM with distinct segment weights (this is called noncrossing harmonic matching). Sakai and Urrutia~\cite{SU2011heavyNM} considered the weight of a segment to be the minimum weight of the two endpoints, and they studied lower and upper bounds for the maximum weighted NM.

In pure mathematics, NM has been studied as a tool to understand the representation theory of groups~\cite{Armstrong2013Rep2,Patrias2022Rep}. A tuple of NMs (a so-called necklace) satisfying a specific property is used to study the topology of harmonic algebraic curves associated with a polynomial over $\C$~\cite{Savitt2009polynomials_necklace}. 
Extremal graph problems where an NM of size $k$ plays the role of a forbidden subgraph are studied in~\cite{Aichholzer2010NMasSubgraph,Gyarfas2011Extremal}.

Besides the application in image processing and circuit design, as mentioned in the introduction, 
NM has also found other applications. One major application is in computational biology. A restricted version of NM (the points are all on a circle), and $k$-noncrossing matchings (no $k$ pairwise intersecting segments, which reduces to the standard NM when $k=2$) have been studied to understand RNA structures~\cite{Bafna1995RNA,Chen2009RNA,Vladimirov2013RNA}. 
In applications that are related to visibility problems (such as robotics) and geometric shape matching, one replaces all or a subset of the points in question by geometric objects~\cite{Rappaport2002visibilitymatching,Aloupis2015objects}. For example, when the question is to match objects to objects, then a segment $pq$ between two objects $A$ and $B$ can be formed by choosing arbitrary points $p$ from $A$ and $q$ from $B$, conditioned on the segment $pq$ not crossing other objects.

The online NM has only been studied very recently. Bose et al.~\cite{bose2020non} initiated the study of online (unweighted) NM and showed that the competitive ratio of deterministic algorithms is $2/3$, while Kamali et al.~\cite{Kamali2022randomizedNM} gave a randomized algorithm with expected competitive ratio approximately $0.6695$. The online bichromatic NM has also been studied in~\cite{bose2020non,sajadpour2021non}.  Boyar et al.~\cite{boyar2024online} initiated the study of weighted online noncrossing matching where weights are put on points; among other results, they gave a randomized algorithm that matches in expectation at least $1/3$ fraction of all points even when points have weights, and proved that no randomized algorithm can have competitive ratio greater than $16/17$.  Finally, the advice complexity has also been studied in several works ~\cite{bose2020non,lavasani2023,boyar2024online}.

\section{Greedy in the \texorpdfstring{$k$}{k}-crossing model}\label{sec:k-crossing}

\subsection{Determine $\beta_k(\ell)$}

We first use Euler's formula for planar graphs to establish a cell-counting formula for arrangements of segments in terms of the total number of intersections.

\begin{lemma}   \label{lem:counting_cells}
    Let $\cL$ be a family of $\ell$ line segments with $m$ total intersections.  Then, the arrangement creates at most
        $\ell+m+1$
cells inside a sufficiently large bounding box.
\end{lemma}

\begin{figure}[ht!]
\centering
\begin{tikzpicture}[
    x=0.60cm,
    y=0.60cm,
    segment/.style={line width=0.8pt},
    extension/.style={line width=0.7pt,dashed},
    point/.style={circle,fill=black,inner sep=1.5pt},
    every node/.style={font=\small}
]
    \draw[line width=0.8pt] (0,0) rectangle (10,6);

    \coordinate (t-left-boundary) at (0,1);
    \coordinate (t-left) at (1.8,2);
    \coordinate (t-right) at (5.4,4);
    \coordinate (p) at (7.2,5);
    \coordinate (q) at (9,6);

    \coordinate (s-bottom-boundary) at (6.2,0);
    \coordinate (s-bottom) at (6.4,1);
    \coordinate (s-top) at (6.8,3);

    \coordinate (u-left-boundary) at (0,4.6);
    \coordinate (u-left) at (2,4);
    \coordinate (u-right) at (8,2.2);
    \coordinate (u-right-boundary) at (10,1.6);

    \coordinate (a) at (2.7,2.5);
    \coordinate (r-top) at (3.0,2.0);
    \coordinate (r-bottom) at (3.9,0.5);
    \coordinate (b) at (4.2,0);

    \coordinate (tr-crossing) at (4.2078,3.3377);
    \coordinate (sr-crossing) at (6.7170,2.5849);

    \draw[extension] (t-left-boundary) -- (t-left);
    \draw[extension] (t-right) -- (q);
    \draw[extension] (s-bottom-boundary) -- (s-bottom);
    \draw[extension] (s-top) -- (p);
    \draw[extension] (u-left-boundary) -- (u-left);
    \draw[extension] (u-right) -- (u-right-boundary);
    \draw[extension] (a) -- (r-top);
    \draw[extension] (r-bottom) -- (b);

    \draw[segment] (t-left) -- (t-right);
    \draw[segment] (s-bottom) -- (s-top);
    \draw[segment] (u-left) -- (u-right);
    \draw[segment] (r-top) -- (r-bottom);

    \node[point] at (0,0) {};
    \node[point] at (0,6) {};
    \node[point] at (10,0) {};
    \node[point] at (10,6) {};

    \node[point] at (t-left-boundary) {};
    \node[point] at (t-left) {};
    \node[point] at (t-right) {};
    \node[point] at (s-bottom-boundary) {};
    \node[point] at (s-bottom) {};
    \node[point] at (s-top) {};
    \node[point] at (u-left-boundary) {};
    \node[point] at (u-left) {};
    \node[point] at (u-right) {};
    \node[point] at (u-right-boundary) {};
    \node[point] at (r-top) {};
    \node[point] at (r-bottom) {};
    \node[point] at (tr-crossing) {};
    \node[point] at (sr-crossing) {};

    \node[point] at (p) {};
    \node[point] at (q) {};
    \node[point] at (a) {};
    \node[point] at (b) {};

    \node[above left=2pt] at (4.8,3.67) {$t$};
    \node[right=3pt] at (6.5,1.5) {$s$};
    \node[right=3pt] at (3.45,1.25) {$r$};
    \node[below right=2pt] at (p) {$p$};
    \node[above=3pt] at (q) {$q$};
    \node[above left=2pt] at (a) {$a$};
    \node[below=3pt] at (b) {$b$};
\end{tikzpicture}
\caption{$a, p, q$ are boundary points for segments $r, s, t$ respectively, while $a,p$ are both meet points of segment $t$. The edges of the graph include both the solid and dashed edges.}
\label{fig:boundary-and-meet-points}
\end{figure}
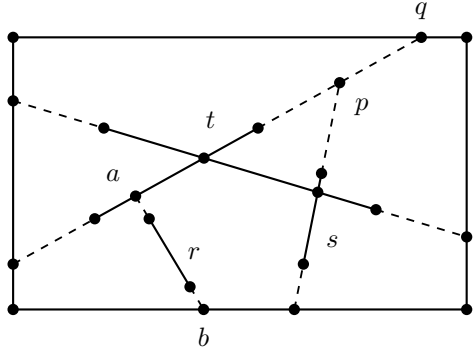

\begin{proof}
    The idea is to construct a planar graph by properly extending the segments and then applying Euler's formula. Let $B$ be  a sufficiently large rectangular bounding box, chosen generically so that no extended segment meets a corner of $B$. Extend each segment in both directions until it \emph{meets} the boundary of $B$ or other (extended) segments, see Figure \ref{fig:boundary-and-meet-points} for an illustration. When an extended segment $s$ meets either the boundary of $B$ or another segment $t$ at a point $p$, we say that $p$ is a boundary point of $s$; in the latter case, $p$ is called a meet point of $t$. This creates a connected planar graph. Next, we count the number of vertices, edges and faces.

The vertices of this graph are: (1) $m$ original intersection points; (2) meet points, which we divide into two types. Let $A$ denote the set of meet points that are generated when a segment meets another segment, and let $C$ denote the set of meet points on the boundary of $B$. We have $|A|+|C| = 2\ell$ because all meet points are also all the boundary points of segments. So, if we let $x = |A|$, then $|C| = 2\ell - x$; (3) the $4$ corner points of $B$.
Thus, the total number of vertices is
\[
    V=m +2\ell+4.
\]

Now, we count edges. For the $i$-th extended segment, let $k_i$ be the number of intersection points and meet points of type $A$ on it, that is, we do not include its two boundary points. These divide the segment into $k_i+1$ edges. Hence, the segments contribute
    $\sum_{i=1}^\ell (k_i+1)=\ell +\sum_{i=1}^\ell k_i$
edges. Since each intersection point lies on exactly two segments, and each meet point belongs to exactly one segment (e.g., $a$ is a meet point of segment $t$, \emph{not} a meet point of segment $r$, in Figure \ref{fig:boundary-and-meet-points}), by double counting we have 
    $\sum_{i=1}^\ell k_i=2m+x$.
Hence, the segments contribute $\ell+2m+x$ edges.

The boundary of $B$ is subdivided by the $2\ell-x$ boundary points, yielding $2\ell-x+4$ additional boundary edges. Thus, the total number of edges is
\[
    E=(\ell+2m+x)+(4+2\ell-x)=3\ell+2m+4.
\]

The graph is connected, so Euler's formula $V-E+F=2$ applies.   The exterior of $B$ is one face, so the number of cells inside $B$ is $F-1$. Therefore, 
    $F-1
    =E-V+1
    =\ell+m+1$.
\end{proof}

\begin{proposition}\label{prop:k-blocker}
    $\beta_k(\ell) \le \ell(1+k/2)+1$, and this bound is asymptotically attained.
\end{proposition}

\begin{proof}
    Let $\cL$ be a family of segments satisfying $k$-crossing property. By double counting the point-segment incidences, the number of intersections is at most $k\ell/2$. Applying Lemma \ref{lem:counting_cells}, we obtain that the number of cells is at most $\ell + k\ell/2 + 1$. We first justify that every cell of this arrangement is convex, and hence can contain at most one blocked point. Since every extension is stopped as soon as it meets an obstacle, every end of every segment lies either on the boundary of $B$ or in the interior of another segment, where it terminates (at a meet point one segment passes straight through while the other one ends, cf.\ Figure~\ref{fig:boundary-and-meet-points}). Hence no cell boundary makes a reflex turn, and no cell is bounded twice by segments lying on the same supporting line: a cell containing points on both sides of one of its boundary segments would have to go around an end of that segment, but every end is supported, either by $\partial B$ or by the segment on which it terminates. Therefore every cell is an intersection of open half-planes with $B$, and is so convex. In particular, two points lying in the same cell are joined by a segment that remains inside the cell, and thus crosses no segment of $\cL$; so the points of a blocked set occupy pairwise distinct cells, and the upper bound for $\beta_k(\ell)$ follows. 

    For the lower bound, we distinguish two cases. For $k=0$ the assertion is exactly the noncrossing case of~\cite{bose2020non}; our construction below therefore only concerns $k\ge1$. Consider the following construction: the grid formed by $k$ horizontal segments and $k$ vertical segments gives a family of segments satisfying $k$-crossing property. Note that it blocks $(k+1)^2$ points. One can then recurse this construction: in each cell one adds a copy of this $k\times k$ construction. This yields the effect of adding $(k+1)^2 -1 = k^2 + 2k = 2k(1+k/2)$ new blocked points by adding $2k$ segments. Hence, in the limit, the number of blocked points would be $1+k/2$ times the number of segments, as desired.
\end{proof}

\subsection{Implication for Greedy}

Theorem \ref{thm:block-to-greedy} together with Proposition \ref{prop:k-blocker} immediately yields the following.

\begin{theorem}\label{thm:k-greedy}
    For every fixed $k\geq0$, the exact asymptotic competitive ratio of
Greedy for the online $k$-crossing matching is $4/(k+6)$.
\end{theorem}

Setting $k=0$, this recovers the competitive ratio of Greedy for online noncrossing matching established previously by Bose et al.~\cite{bose2020non}.

\section{Greedy in the triangle-free-crossing model}\label{sec:triangle-free}

\subsection{A construction for $\beta_\triangle(\ell)$} \label{sec:lowerbound_triangle}

Let $\ell,p$ be two parameters.  We say a pair $(\ell,p)$ is \emph{admissible} if there is a triangle-free-crossing family of $\ell$ segments blocking $p$ points, equivalently, $\beta_\triangle(\ell) \ge p$.

\begin{figure}[ht!]
\centering
\includegraphics[width=.4\linewidth,trim=0 24 0 0,clip]{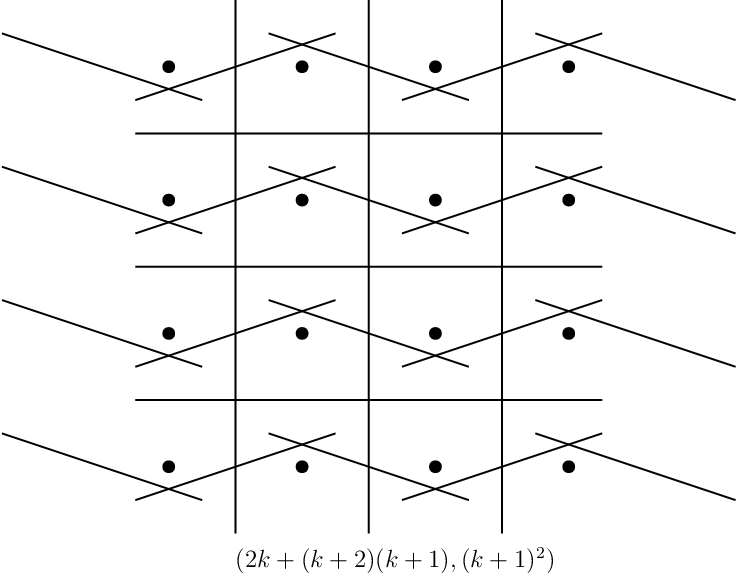}
\caption{$(2h+(h+2)(h+1),(h+1)^2)$ is admissible; the figure shows the case $h=3$.}
\label{fig:triangle-grid}
\end{figure}

\begin{proposition}\label{prop:triangle-blocker}
    $\ell - O(\sqrt{\ell}) \le \beta_\triangle(\ell) \le  \frac{\ell^2 +\ell +2}{2}$.
\end{proposition}

\begin{proof}
	Since the $\ell$ supporting lines can determine at most ${\ell \choose 2} + \ell + 1 = \frac{\ell^2 +\ell +2}{2}$ convex regions (see  \cite[Problem 3.2]{ECG_Szilard}), and clearly each convex region can contain at most one blocked point, this shows the upper bound. 

	For the lower bound, we exhibit an explicit construction, see an illustration in Figure \ref{fig:triangle-grid}. The construction is to  give  $h$ horizontal and vertical segments forming a grid, and in each row add $h+2$ segments in a wave form as indicated in the illustration for $h=3$.  One can directly verify that the family of segments is triangle-free, and all the black points are indeed blocked by the configuration. There are $2h + (h+2)(h+1) = h^2 + 5h + 2$ line segments and $(h+1)^2 = h^2 + 2h + 1$ blocked points, i.e., $(2h + (h+2)(h+1), (h+1)^2)$ is admissible for every $h \ge 1$. This implies for $\ell$ of the form $\ell = 2h + (h+2)(h+1)$, one has
        $\beta_\triangle(\ell) \ge \ell - O(\sqrt{\ell})$.
For $\ell$ that is not of the form $2h + (h+2)(h+1)$, one may consider the $h$ for which
    $2h + (h+2)(h+1) \le \ell < 2(h+1) + (h+3)(h+2)$.
For such $\ell$, one also has $(\ell, (h+1)^2)$ is admissible because trivially $\beta_\triangle(\ell)$ is non-decreasing. Hence, the claimed lower bound holds for every $\ell$.
\end{proof}

We remark that the elegant construction in Figure \ref{fig:triangle-grid} comes from dogged work of trial and error based on the  observation of the smaller configuration in Figure \ref{fig:triangle-blocking}.

\subsection{Implication for Greedy}

Theorem \ref{thm:block-to-greedy} together with Proposition \ref{prop:triangle-blocker} immediately yields the following.

\begin{theorem}\label{thm:triangle-free-greedy}
    The asymptotic competitive ratio of Greedy for the online triangle-free-crossing matching is at most $2/3$.
\end{theorem}

\subsection{Two-layer-Greedy}

In this section we introduce a simple deterministic two-layer-Greedy algorithm that achieves an asymptotic competitive ratio of $4/5 > 2/3$ for online triangle-free-crossing matching. The idea is to maintain a partition
$\cL=\cL_1\mathbin{\dot\cup}\cL_2$ in which each layer is noncrossing. When
an online point $v$ arrives, it scans the earlier unmatched points in arrival order.
If the scan finds a point $u$ such that $uv$ is noncrossing within $\cL_1$, it matches $u$ and $v$ and adds $uv$ to $\cL_1$.
Otherwise, it performs the same process for $\cL_2$.  If neither scan succeeds, it leaves $v$ unmatched.  Since both $\cL_1$ and $\cL_2$ are noncrossing, $G(\cL)$ is bipartite and therefore
triangle-free.  

\begin{theorem}\label{thm:two-layer}
    The asymptotic competitive ratio of the two-layer-Greedy is exactly $4/5$.
\end{theorem}

\begin{proof}
    Let $\cP$ be its unmatched point set, put
$p=|\cP|$, and write $\ell_i=|\cL_i|$ and
$\ell=\ell_1+\ell_2$. By the property of the two-layer-Greedy,  each layer blocks $\cP$ with respect to the noncrossing property.  Applying Proposition \ref{prop:k-blocker} with $k=0$ separately to the two layers gives
    $p \le \ell_i + 1$
for $i=1,2$.
    Since $n=2\ell + p = 2(\ell_1 + \ell_2) +p$, a direct simple calculation gives the lower bound.

    To see the tightness, consider the following construction. Let $\cL_1$ be a noncrossing family of $\ell$ segments blocking a set $\cP$ of $\ell+1$ points. By Proposition \ref{prop:k-blocker}, such a configuration exists. Now, rotate each segment of $\cL_1$ slightly about its own midpoint, with the rotations sufficiently small to preserve all required and forbidden crossings, to obtain another configuration $\cL_2$. Clearly, $\cL_2$ is also noncrossing, and by the construction, a segment $s \in \cL_1$ necessarily intersects its counterpart $s' \in \cL_2$, and does not intersect any other segment in both $\cL_1$ and $\cL_2$. In other words, $\cL = \cL_1 \cup \cL_2$ is triangle-free-crossing. Clearly, $\cL_2$ also blocks $\cP$ (under the noncrossing property). Now, we claim that $\cL$ as a triangle-free-crossing family blocks $\cP$ under the triangle-free-crossing property. Indeed, any two points $p_1,p_2 \in \cP$ necessarily intersect a segment $s \in \cL_1$, and hence will also intersect the segment $s' \in \cL_2$ because $s'$ is only a small perturbation of $s$, and we know segment $s$ and $s'$ are intersecting, hence, $p_1p_2$ and $s$ and $s'$ form a triangle. Since $|\cP| = \ell+1$ and $|\cL| = 2\ell$, this proves the asymptotic upper bound.
\end{proof}

\subsection{Vertex cover number and explicit constructions for small Ramsey numbers}	\label{sec:constructions}

We proved a lower bound for $\beta_\triangle$ by an explicit construction in Section \ref{sec:lowerbound_triangle}. Below we provide another possible route that is via the vertex cover number of a graph; this also leads to an interesting possible connection to Ramsey theory. These connections might become more essential when one considers a generalization of the triangle-free-crossing matching, such as online $K_t$-free-crossing matching for $t \ge 4$. In such a case, an explicit construction like in Figure \ref{fig:triangle-grid} might become much harder. 

Let $\cL$ denote a triangle-free-crossing family of $\ell$ line segments. 
Let $\tau(\cL) = \tau(G(\cL))$ denote the size of a minimum vertex cover of the associated graph $G(\cL)$. Let
\begin{equation}    \label{eq:def_tau-m}
    \tau(\ell) := \max \{\tau(\cL): |\cL| = \ell,\ \cL \text{ is triangle-free-crossing}\}.
\end{equation}

\begin{conjecture}  \label{conj:vtx_conj}
    For every $\ell\in \N$, $\beta_\triangle(\ell) \ge \tau(\ell)$.
\end{conjecture}

A stronger conjecture might be true: for every triangle-free-crossing family $\cL$ of line segments, $\beta_\triangle(\cL) \ge \tau(\cL)$. Our lower bound in Proposition \ref{prop:triangle-blocker}  implies Conjecture \ref{conj:vtx_conj} holds asymptotically. 

\begin{theorem}	\label{thm:conj_true_for_large_m}
	Conjecture \ref{conj:vtx_conj} holds for sufficiently large $\ell$.
\end{theorem}

\begin{proof}
	Since every triangle-free graph on $\ell$ vertices has an independent set of size $\ge \Omega(\sqrt{\ell\log \ell})$ (see, e.g., \cite{kim1995ramsey}), by Proposition \ref{prop:triangle-blocker},  for sufficiently large $\ell$ we have, 
	$\tau(\ell) \le \ell -  \Omega(\sqrt{\ell\log \ell}) \le  \ell - O(\sqrt{\ell})  \le \beta_\triangle(\ell)$.
\end{proof}

One reason to consider the graph parameter is that results from there can be applied to help study the Greedy algorithm. Indeed, Walczak \cite{walczak2015triangle} constructed a triangle-free family $\cL$ of $\ell$ line segments for which $\alpha(G(\cL)) = \Theta(\ell/\log\log \ell)$, where $\alpha(G)$ denotes the size of a maximum independent set. Hence, 
        $\tau(G(\cL)) = \ell - \alpha(G(\cL)) = \ell - \Theta(\ell/\log\log \ell) 
        = (1- \Theta(1/\log\log \ell)) \cdot \ell$. 
If Conjecture \ref{conj:vtx_conj} were true, then, 
	$\beta_\triangle(\ell) \ge \tau(\ell) \ge  (1- \Theta(1/\log\log \ell)) \cdot \ell$,
which is also sufficient to imply Theorem \ref{thm:triangle-free-greedy}.

Having established Conjecture \ref{conj:vtx_conj} for large $\ell$, it remains to check it for small $\ell$. Below, we give explicit constructions for cases corresponding to some Ramsey numbers, all supporting Conjecture \ref{conj:vtx_conj}. 

The Ramsey number $R(r,t)$ is the minimum integer $N$ such that every red--blue coloring of the edges of $K_N$ contains either a red clique $K_r$ or a blue clique $K_t$. Alternatively, if a graph on $R(r,t)$ vertices is $K_r$-free, then it must contain an independent set of size at least $t$. We will mainly be interested in $R(3,t)$. By the Ramsey property, for $\ell = R(3,t)$, we have $\tau(\ell) \le \ell-t$. This simple bound for $\tau$ is the reason we choose $\ell = R(3,t)$ as our target family size.

Table \ref{tab:Ramsey} is the current knowledge of $R(3,t)$ for $t \le 10$, see \cite{radziszowski2012small}. The case for $t=10$ is still undetermined, $40\le R(3,10) \le 41$, see \cite{angeltveit2024R_3}. We will establish the lower bound for $\beta_\triangle(\ell)$ in Table \ref{tab:Ramsey}, thus establishing the following.

\begin{table}[ht!]
\centering
\caption{$R(3,t)$ for $t= 3, \ldots, 10$.} \label{tab:Ramsey}
\begin{tabular}{|c|c|c|c|c|c|c|c|c|c|}
\hline
$t$     & 3 & 4 & 5 & 6 & 7 & 8 & 9 & \multicolumn{2}{c|}{10} \\
\hline
$\ell = R(3,t)$ & 6 & 9 & 14 & 18 & 23	& 28 & 36 &	40? & 41? \\
\hline
$\tau(\ell)\le \ell-t$ & 3 & 5 & 9 & 12 & 16	& 20 & 27 &	30 or 31 & 31\\
\hline
$\beta_\triangle(\ell)\ge $ & 4 & 7 & 11 & 14 & 18	& 22 & 28 &	31 & 32 \\
\hline
\end{tabular}
\end{table}

\begin{theorem}
	Conjecture \ref{conj:vtx_conj} holds for $\ell=R(3,t)$ for $t=3,\ldots, 10$.
\end{theorem}

The constructions will use some simple gadgets in Figure \ref{fig:gadgets}, i.e., the admissible small gadgets $(3,2)$, $(4,3)$, $(5,4)$, and $(8,6)$. One can see that $(8,6)$ is the merge of $(3,2)$ and $(5,4)$.

\begin{figure}[ht!]        
    \centering
    \includegraphics[scale=.5]{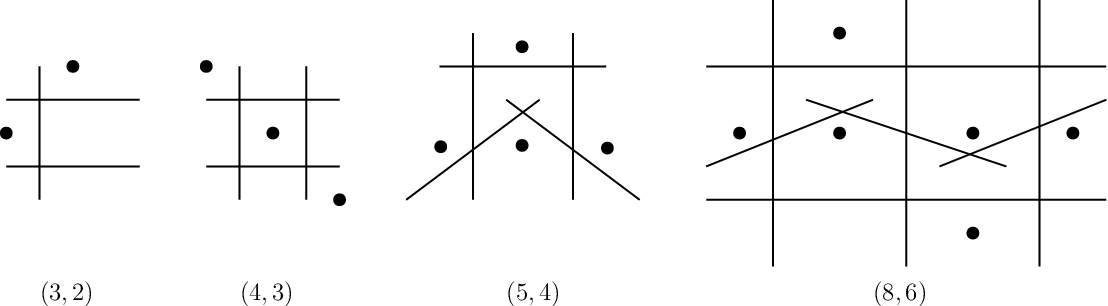}
    \caption{Admissible small gadgets used in constructions}    \label{fig:gadgets}
\end{figure}

\paragraph{The case $t=3$}

In this case, $\ell=R(3,3) = 6$. Since $(5,4)$ is admissible as shown in Figure \ref{fig:gadgets}, $(6,4)$ is also admissible.

\paragraph{The case $t=4$}

In this case, $\ell=R(3,4) = 9$. Figure \ref{fig:gadgets} shows $(8,6)$ is admissible, hence $(9,6)$ is also admissible. We give a better construction in Figure \ref{fig:t4_9} showing that $(9,7)$ is admissible. One can directly check that the construction is triangle-free.  Note that the construction can be viewed as merging $(4,3)$ and $(5,4)$, while removing one cross of blue segments in $(4,3)$ to be triangle-free.

\begin{figure}[ht!]        
    \centering
    \includegraphics[scale=.5]{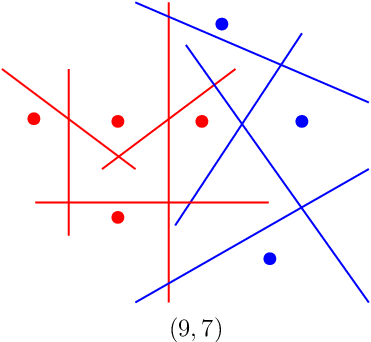}
    \caption{9 segments block 7 points.}    \label{fig:t4_9}
\end{figure}

\paragraph{The case $t=5$}

In this case, $\ell=R(3,5) = 14$. In Figure \ref{fig:t5_14}, the first red and blue block gives a construction that merges $(5,4)$ and $(8,6)$, which results in an admissible pair $(5,4) + (8,6) = (13, 10)$, hence $(14,10)$ is also admissible. One can directly verify that this construction in Figure \ref{fig:t5_14} is triangle-free  and all the points are indeed blocked. In fact, $(14,11)$ is also admissible as shown in Figure \ref{fig:t5and7}, which is obtained by adding the green segment and the green blocked point in the $(13,10)$ construction. 

Furthermore, Figure \ref{fig:t5_14} shows that the admissible pair $(13, 10)$ can copy itself, i.e., $k(13,10)$ is admissible for all $k\ge 1$.

\begin{figure}[ht!]        
    \centering
    \includegraphics[scale=.5]{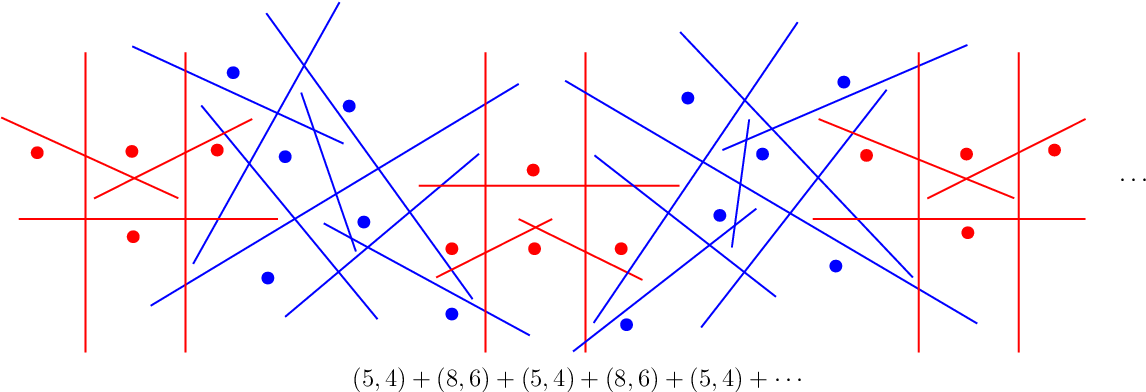}
    \caption{Admissible pairs $k(13,10)$.}    \label{fig:t5_14}
\end{figure}

\begin{figure}[ht!]        
    \centering
    \includegraphics[scale=.5]{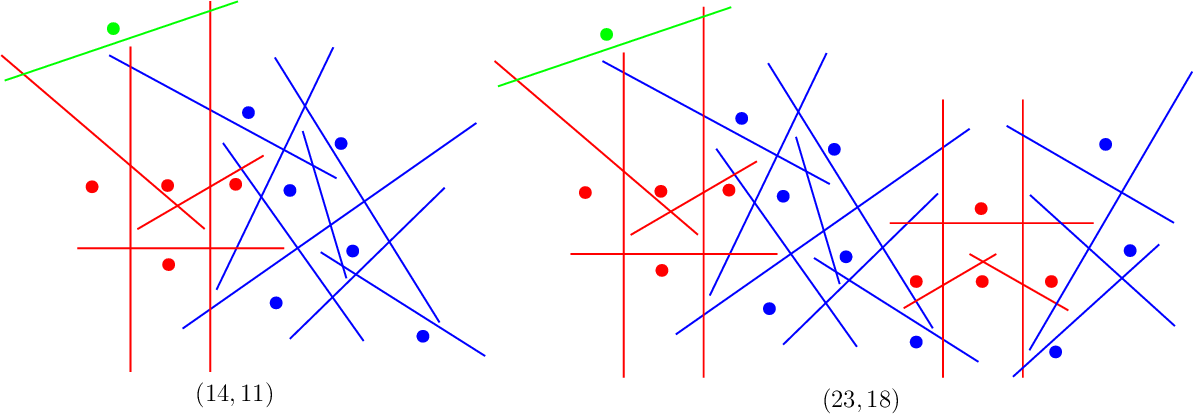}
    \caption{Left: 14 segments block 11 points. Right: 23 segments block 18 points.}    \label{fig:t5and7}
\end{figure}

\paragraph{The case $t=6$}

In this case, $\ell=R(3,6) = 18$. Note that in Figure \ref{fig:t5_14}, one can choose the part $(5,4) + (8,6) + (5,4) = (18,14)$, as desired.

\paragraph{The case $t=7$}

In this case, $\ell=R(3,7) = 23$. On the right of Figure \ref{fig:t5and7},  we merge the construction in Figure \ref{fig:t4_9} and the left of Figure \ref{fig:t5and7}. One can check this construction is triangle-free and all the points are indeed blocked. Hence,
    $(9,7) + (14,11) = (23,18)$ is admissible.

\paragraph{The case $t=8$}

In this case, $\ell=R(3,8) = 28$. Firstly, we have that 
$2\cdot (13,10) = (26,20)$ is admissible. Based on this construction, one can add two segments, adding two blocked points, as shown in Figure \ref{fig:t8_28}. Hence, $(28,22)$ is admissible. Note that one green segment and point are added to the red block $(5,4)$, while the other green segment and point are added to the blue block $(8,6)$.

\begin{figure}[ht!]        
    \centering
    \includegraphics[scale=.5]{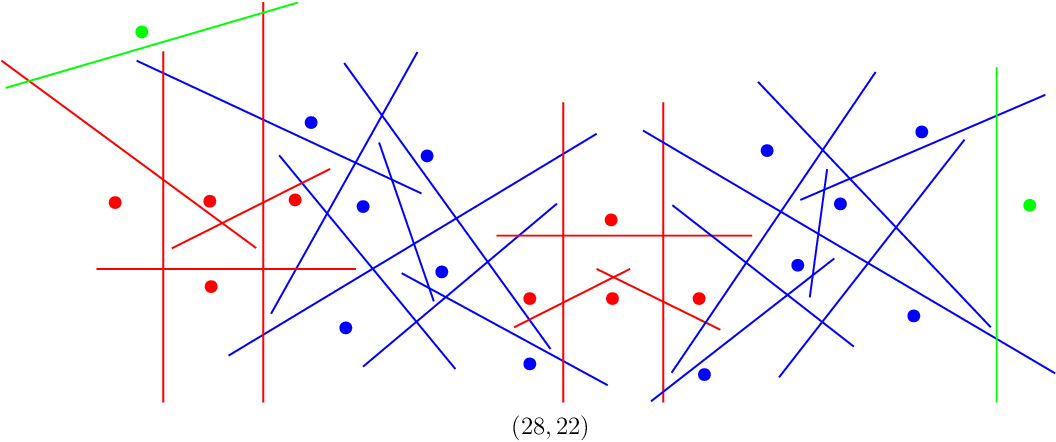}
    \caption{28 segments block 22 points.}    \label{fig:t8_28}
\end{figure}

\paragraph{The case $t=9$}

In this case, $\ell=R(3,9) = 36$. Firstly, we have that
    $3\cdot (13,10) = (39,30)$  is admissible. 
We can remove a $(3,2)$ from it and get that
    $(39,30) - (3,2) = (36,28)$ is admissible, as follows.
Consider the construction in Figure \ref{fig:t5_14}. In the first red block on the left, remove the horizontal segment, the left vertical segment, and the left slanted segment, and remove the leftmost red point and the bottom red point. Observe that the remaining configuration is still a correct construction.

\paragraph{The case $t=10$}

In this case, $40 \le R(3,10)\le 41$. 

Consider $\ell=41$. Since $R(3,10)\le 41$, we know $\tau(41) \le 41-10=31$. Firstly, we have that $3\cdot (13,10) = (39,30)$  is admissible. Similarly, as in Figure \ref{fig:t8_28}, we can add two green segments and two green blocked points, which gives that $(39+2,30+2) = (41,32)$ is admissible.

Consider now $\ell=40$. If $R(3,10) = 40$, then $\tau(40) \le 40-10 = 30$. If $R(3,10) = 41$, then there exists some triangle-free graph on $40$ vertices that does not contain an independent set of size $10$ (but does contain an independent set of size $9$ since $R(3,9) = 36 < 40$), hence, $\tau(40) \le 40-9 = 31$. In both cases, we have $\tau(40) \le 31$. The construction for admissible $(41,32)$ implies that $(40,31)$ is also admissible.

\section{Randomized algorithms in the noncrossing model}\label{sec:noncrossing}

We now turn to the noncrossing model, that is, the case $k=0$ of the k-crossing model, and revisit upper and lower bounds of randomized algorithms.

\subsection{An asymmetric randomized algorithm}     \label{sec:rand_alg}

Without loss of generality we assume all online points lie in a big rectangular region.
Our algorithm maintains a partition into convex regions. A region is partitioned into smaller sub-regions when a line segment appears (as a result of matching two online points inside this region): extend the segment in both directions until it hits the boundary of the region. Each region has type $A$ or $B$, and the root region (the big rectangle) has type $A$. For type-$T$, associate a probability distribution $(x_T,y_T,y_T,z_T)$, whose meaning will be explained shortly, and whose values will be determined later.

An inherited unmatched point, if present, is the first point of the new
region. Let $v_1,v_2,v_3,v_4$ be the first four points assigned to a type-$T$ region.
The local rule is as follows. For notation simplicity, we say match $v_iv_j$ to mean match the point $v_i$ to $v_j$.

\begin{enumerate}
    \item At the arrival of $v_2$, match $v_1v_2$ with probability $x_T$.
    
    \item If $v_1v_2$ is not matched, then at the arrival of $v_3$ match
          $v_1v_3$ or $v_2v_3$, each with probability $y_T/(1-x_T)$. 
    
    \item Suppose one segment among $v_1,v_2,v_3$ has been matched and $u$ is the
          remaining point. At the arrival of $v_4$, match $uv_4$ deterministically when
          it does not cross the existing segment.
    
    \item If no segment has been matched before $v_4$, match $v_4$ deterministically to a point $u \in \{v_1,v_2,v_3\}$ so that it separates the two unmatched points on different sides.
\end{enumerate}

Now, the intuitive meaning of the distribution $(x_T,y_T,y_T,z_T)$ should be clear: $x_T,y_T,y_T$  are respectively the probabilities that segments $v_1v_2, v_1v_3, v_2v_3$ are matched, and $z_T$ is the probability that none of them is matched. 

There are two outcomes of the four points, see an illustration in Figure \ref{fig:nc-partitions}. A \emph{single-pair outcome} leaves two points on
opposite sides of the supporting line of the matched segment; the two sub-regions are assigned type $B$. A \emph{double-pair outcome} matches all four points, resulting in three sub-regions that are assigned type $A$. The algorithm then repeats the above rules in each sub-region, using independent random bits. Clearly, the algorithm correctly outputs a noncrossing matching.

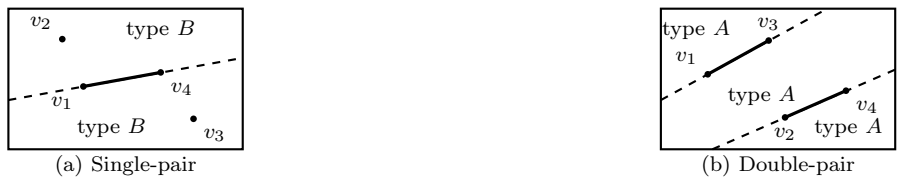
\begin{figure}[ht!]
\centering
\begin{minipage}[t]{.48\linewidth}
\centering
\begin{tikzpicture}[scale=.62,every node/.style={font=\scriptsize}]
  \draw[thick] (0,0) rectangle (5,3);
  \draw[dashed,thick] (0,1.05)--(5,1.95);
  \draw[very thick] (1.6,1.34)--(3.25,1.64);
  \fill (1.6,1.34) circle (2pt) node[below left] {$v_1$};
  \fill (3.25,1.64) circle (2pt) node[below right] {$v_4$};
  \fill (1.15,2.35) circle (2pt)
      node[above left] {$v_2$};
  \fill (3.95,.65) circle (2pt)
      node[below right] {$v_3$};
  \node at (3.25,2.55) {type $B$};
  \node at (2.2,.45) {type $B$};
  \node at (2.5,-.35) {(a) Single-pair};
\end{tikzpicture}
\end{minipage}\hfill
\begin{minipage}[t]{.48\linewidth}
\centering
\begin{tikzpicture}[scale=.62,every node/.style={font=\scriptsize}]
  \draw[thick] (0,0) rectangle (5,3);
  \draw[dashed,thick] (0,1.05)--(3.55,3);
  \draw[dashed,thick] (1.1,0)--(5,1.7);
  \draw[very thick] (1,1.6)--(2.3,2.32);
  \draw[very thick] (2.65,.68)--(3.95,1.25);
  \fill (1,1.6) circle (2pt) node[above left] {$v_1$};
  \fill (2.3,2.32) circle (2pt) node[above] {$v_3$};
  \fill (2.65,.68) circle (2pt) node[below] {$v_2$};
  \fill (3.95,1.25) circle (2pt) node[below right] {$v_4$};
  \node at (.75,2.55) {type $A$};
  \node at (2.15,1.15) {type $A$};
  \node at (4,.45) {type $A$};
  \node at (2.5,-.35) {(b) Double-pair};
\end{tikzpicture}
\end{minipage}
\caption{The recursive partition rule.  Dashed segments indicate the extensions used as new boundaries.}
\label{fig:nc-partitions}
\end{figure}

Our algorithm can be viewed as an asymmetric version of Kamali, Nikbakht, and Sajadpour~\cite{Kamali2022randomizedNM}: in their algorithm the sub-region has no type and $y$ and $z$ are equal; specifically, their algorithm uses uniform probability $x=1/2$ and $y=z=1/6$. The new aspects of our algorithm are that we let the probability depend on the types of sub-regions depending on whether a single-pair or a double-pair has been matched, and secondly, we allow $y$ and $z$ to be unequal. The analysis of the performance of the algorithm follows similar steps as in ~\cite{Kamali2022randomizedNM}.  For completeness, we give the details of the algorithm and its analysis.

\subsubsection{A structure lemma}

 For the analysis, we say a single-pair outcome is
\emph{good} if both children (i.e., the sub-regions) contain at least two points, including their inherited points. A double-pair outcome is \emph{good} if at least one of its three
children is empty. Otherwise, we say they are \emph{bad}.
For a type-$T$ region, we call the event that $v_1v_2$ is matched the
$x_T$-branch, and we call the events that $v_1v_3$ and $v_2v_3$ are matched
the two $y_T$-branches. The remaining event, in which no segment is matched
before $v_4$, is the $z_T$-branch. These four branches have probabilities
$x_T,y_T,y_T,z_T$, respectively.

\begin{figure}[ht!]
\centering
\begin{minipage}[t]{.47\linewidth}
\centering
\begin{tikzpicture}[scale=.78,every node/.style={font=\scriptsize}]
  \coordinate (v1) at (-1.5,-.9);
  \coordinate (v2) at (1.5,-.9);
  \coordinate (v3) at (0,1.35);
  \coordinate (v4) at (.08,-.05);

  \draw (v1)--(v2)--(v3)--cycle;
  \draw[very thick] (v1)--(v2);
  \draw[very thick] (v4)--(v3);
  \draw[dashed] ($(v4)!-.65!(v3)$)--($(v4)!1.7!(v3)$);

  \fill (v1) circle (2pt) node[below left] {$v_1$};
  \fill (v2) circle (2pt) node[below right] {$v_2$};
  \fill (v3) circle (2pt) node[left] {$v_3$};
  \fill (v4) circle (2pt) node[right] {$v_4$};

  \fill (-.55,.25) circle (1.8pt);
  \fill (.55,.28) circle (1.8pt);

  \node[align=center] at (0,-1.45) {(a) $v_4$ is interior};
\end{tikzpicture}
\end{minipage}\hfill
\begin{minipage}[t]{.47\linewidth}
\centering
\begin{tikzpicture}[scale=.78,every node/.style={font=\scriptsize}]
  \coordinate (v1) at (-1.5,-.9);
  \coordinate (v2) at (1.5,-.9);
  \coordinate (v4) at (0,1.35);
  \coordinate (v3) at (.1,-.05);

  \draw (v1)--(v2)--(v4)--cycle;
  \draw[very thick] (v1)--(v2);
  \draw[very thick] (v3)--(v4);
  \draw[dashed] ($(v3)!-.65!(v4)$)--($(v3)!1.7!(v4)$);

  \fill (v1) circle (2pt) node[below left] {$v_1$};
  \fill (v2) circle (2pt) node[below right] {$v_2$};
  \fill (v4) circle (2pt) node[left] {$v_4$};
  \fill (v3) circle (2pt) node[right] {$v_3$};

  \fill (-.55,.3) circle (1.8pt);
  \fill (.55,.28) circle (1.8pt);

  \node[align=center] at (0,-1.45) {(b) $v_3$ is interior};
\end{tikzpicture}
\end{minipage}
\caption{The two triangular-hull cases. Unlabeled points represent later
points. A bad double pair supplies, in particular, the two depicted
later-point witnesses on opposite sides of the dashed line, making the
corresponding single pair good.}
\label{fig:nc-triangular-cases}
\end{figure}
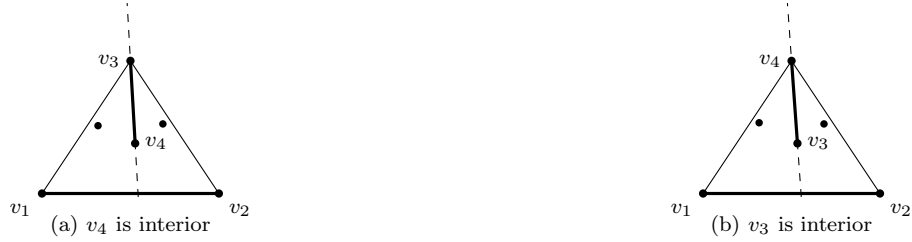

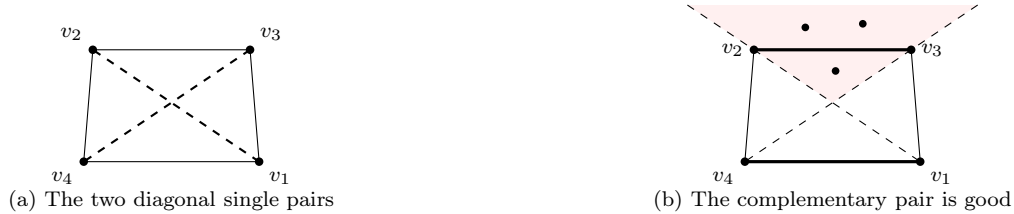
\begin{figure}[ht!]
\centering
\begin{minipage}[t]{.47\linewidth}
\centering
\begin{tikzpicture}[scale=.8,every node/.style={font=\scriptsize}]
  \coordinate (v2) at (-1.3,.95);
  \coordinate (v3) at (1.3,.95);
  \coordinate (v1) at (1.45,-.9);
  \coordinate (v4) at (-1.45,-.9);

  \draw (v2)--(v3)--(v1)--(v4)--cycle;
  \draw[dashed,thick] (v1)--(v2);
  \draw[dashed,thick] (v3)--(v4);

  \fill (v1) circle (2pt) node[below right] {$v_1$};
  \fill (v2) circle (2pt) node[above left] {$v_2$};
  \fill (v3) circle (2pt) node[above right] {$v_3$};
  \fill (v4) circle (2pt) node[below left] {$v_4$};

  \node[align=center] at (0,-1.55)
       {(a) The two diagonal single pairs};
\end{tikzpicture}
\end{minipage}\hfill
\begin{minipage}[t]{.47\linewidth}
\centering
\begin{tikzpicture}[scale=.8,every node/.style={font=\scriptsize}]
  \coordinate (v2) at (-1.3,.95);
  \coordinate (v3) at (1.3,.95);
  \coordinate (v1) at (1.45,-.9);
  \coordinate (v4) at (-1.45,-.9);

  \coordinate (o) at (0,.0754545);

  \coordinate (lefttop) at ($(v2)!-.4!(v1)$);
  \coordinate (righttop) at ($(v3)!-.4!(v4)$);

  \fill[red!6] (o)--(lefttop)--(righttop)--cycle;

  \draw (v2)--(v3)--(v1)--(v4)--cycle;
  \draw[dashed] (lefttop)--(v1);
  \draw[dashed] (righttop)--(v4);
  \draw[very thick] (v2)--(v3);
  \draw[very thick] (v4)--(v1);

  \fill (v1) circle (2pt) node[below right] {$v_1$};
  \fill (v2) circle (2pt) node[left] {$v_2$};
  \fill (v3) circle (2pt) node[right] {$v_3$};
  \fill (v4) circle (2pt) node[below left] {$v_4$};

  \fill (-.45,1.32) circle (1.8pt);
  \fill (.05,0.6) circle (1.8pt);
  \fill (.5,1.38) circle (1.8pt);

  \node[align=center] at (0,-1.55)
       {(b) The complementary pair is good};
\end{tikzpicture}
\end{minipage}
\caption{The convex-position case. Unlabeled points represent later points.
If both diagonal single pairs are bad, all later points lie in one
quarter-plane (shaded), so the complementary boundary pairing is a good
double pair.}
\label{fig:nc-convex-case}
\end{figure}

\begin{lemma}\label{lem:four-point}
For the first four points $v_1,v_2,v_3,v_4$ of a region, the following cases
exhaust all possibilities.
\begin{enumerate}[(1)]
\item The $z_T$-branch always produces a single-pair outcome.
\item If the convex hull of the four points is a triangle, the $x_T$-branch
      and both $y_T$-branches produce double-pair outcomes. Moreover, either
      the $z_T$-branch is a good single-pair outcome or one of the other three
      branches is a good double-pair outcome.
\item If the four points are in convex position, exactly one among the
      $x_T$-branch and the two $y_T$-branches produces a single-pair outcome.
      Either this branch or the $z_T$-branch is a good single-pair outcome, or
      one of the other two branches is a good double-pair outcome.
\end{enumerate}
\end{lemma}

\begin{proof}
(1) In the $z_T$-branch, no segment is matched among $v_1,v_2,v_3$. The last
step of the local rule matches exactly one segment incident to $v_4$, so this
branch always produces a single pair.

(2) Suppose first that the convex hull of the four points is a triangle. Then,
every perfect pairing is noncrossing. Hence, the $x_T$-branch and both
$y_T$-branches produce double pairs. Let $e=v_iv_4$ be the segment selected
in the $z_T$-branch, where $\{i,j,k\}=\{1,2,3\}$. Its supporting line
separates $v_j$ and $v_k$, and the branch whose first segment is $v_jv_k$
completes its matching with $e$. Figure~\ref{fig:nc-triangular-cases} uses the
representative labeling $e=v_3v_4$: panel~(a) shows the case in which $v_4$
is the interior point, while panel~(b) shows, up to symmetry, the case in
which an earlier point, represented by $v_3$, is interior.

If this double pair is good, the claim follows. Otherwise, each of its three
children contains a later point. In particular, a later point lies on each
side of the dashed supporting line of $e$, as illustrated in both panels of
Figure~\ref{fig:nc-triangular-cases}. Each side also contains one of the
unmatched points $v_j,v_k$ inherited by the single-pair outcome $e$.
Therefore, both children of this single pair have size at least two, and the
$z_T$-branch is good.

(3) Suppose now that the four points are in convex position and are labeled as in
Figure~\ref{fig:nc-convex-case}. Exactly one of the three perfect pairings
crosses. Therefore, exactly one among the
$x_T$-branch and the two $y_T$-branches produces only one pair. The two
single pairs are the diagonals $v_1v_2$ and $v_3v_4$, shown dashed in
Figure~\ref{fig:nc-convex-case}(a).

If either diagonal single pair is good, the claim follows. If both are bad,
all later points lie in one open quarter-plane determined by their supporting
lines. After a reflection if necessary, this is the shaded quarter-plane in
Figure~\ref{fig:nc-convex-case}(b). One of the two boundary pairings has a
child outside this quarter-plane; the illustrated pairing is
$\{v_2v_3,v_1v_4\}$. That child contains no later point and is therefore
empty. Hence, the corresponding branch is a good double-pair outcome.
\end{proof}

\subsubsection{The optimization problem}

For $T\in\{A,B\}$, let $f_T(r)$ be the supremum, over all fixed $r$-point
inputs assigned to a type-$T$ region, of the expected number of unmatched
points. An inherited unmatched point, when present, is included in this
count. We set $f_T(0)=0$. For example, $f_T(1)=1$ and
$f_T(2) = 2(1 - x_T)$.
The goal is to find constants $c,d_A,d_B$, preferably small $c$, such that
\begin{equation}\label{eq:nc-linear-target}
 f_T(r)\leq cr+d_T,
 \qquad \forall\ T\in\{A,B\},\ r\geq2.
\end{equation}
We will first describe sufficient constraints on the variables
$x_T,y_T,z_T$ and $c,d_T$ for $T\in\{A,B\}$, then we show that a specific
choice of numerical values for these variables gives the improved performance
of our algorithm.

\medskip
\noindent\textbf{(A) The constraints on distributions.}

We require the probability variables to satisfy
\begin{equation}\label{eq:nc-system-prob}
 x_T+2y_T+z_T=1,\qquad x_T\geq y_T\geq0,\qquad z_T\geq0,
 \qquad \forall\ T\in\{A,B\}.
\end{equation}
Here, we require $x_T\geq y_T$ for easier analysis. In particular, this
condition implies that every branch that matches a segment before $v_4$ has
probability at least $y_T$. Among these three branches, at most one produces
a single-pair outcome. Thus, at least two produce double-pair outcomes, and
their total probability is at least $2y_T$.

\medskip
\noindent\textbf{(B) The constraints on the first seven points.}

We already saw $f_T(1)=1$ and $f_T(2)=2(1-x_T)$. 
On a three-point input, an earlier segment is matched with probability
$1-z_T$, leaving one point unmatched, whereas the $z_T$-branch leaves all
three points unmatched. Hence,
\[
 f_T(3)=1(1-z_T)+3z_T=1+2z_T.
\]

For four points, let $\delta_T$ denote the probability of a double-pair
outcome. At most one of the three branches that matches a segment before
$v_4$ can correspond to the crossing perfect pairing. Since each of the other
two branches has probability at least $y_T$, we have
$\delta_T\geq2y_T$.  Consequently,
\[
 f_T(4)
 =2(1-\delta_T)+0\cdot\delta_T \leq2(1-2y_T).
\]

Since a single-pair outcome creates type-$B$ children, define
\begin{equation}\label{eq:nc-q234}
 q_2=2(1-x_B),\qquad
 q_3=1+2z_B,\qquad
 q_4=2(1-2y_B).
\end{equation}
Thus, $f_B(i)\leq q_i$ for $i=2,3,4$. Introduce auxiliary variables
$L_6,L_7,H_7$ satisfying
\begin{equation}\label{eq:nc-base-epigraph}
\begin{aligned}
 L_6&\geq1+q_3, & L_6&\geq2q_2,\\
 L_7&\geq1+q_4, & L_7&\geq q_2+q_3,\\
 H_7&\geq L_7,  & H_7&\geq3.
\end{aligned}
\end{equation}

We next calculate the bounds for five, six, and seven points. For five
points, a single-pair outcome has two type-$B$ children of sizes $(1,2)$, so
its cost (i.e., expected unmatched points) is at most
 $f_B(1)+f_B(2)\leq1+q_2$.
A double-pair outcome matches four points and leaves one point unmatched, so
its cost is one. Using $\delta_T\geq2y_T$ and $q_2\geq0$, we obtain
\[
 f_T(5)
 \leq(1-\delta_T)(1+q_2)+\delta_T
 \leq (1-2y_T)(1+q_2)+2y_T.
\]

For six points, the two type-$B$ child sizes of a single-pair outcome sum to
four. The possible patterns are $(1,3)$ and $(2,2)$, whose costs are at most
$1+q_3$ and $2q_2$, respectively. Thus, the single-pair cost is at most
$L_6$. The three type-$A$ child sizes of a double-pair outcome sum to two, so
the possible nonzero patterns are $(2)$ and $(1,1)$, both of cost at most
two. Moreover, $L_6\geq1+q_3\geq2$. Therefore,
\[
 f_T(6) \leq(1-\delta_T)L_6+2\delta_T
 \leq (1-2y_T)L_6+4y_T.
\]

For seven points, the two type-$B$ child sizes of a single-pair outcome sum
to five. The possible patterns are $(1,4)$ and $(2,3)$, whose costs are at
most $1+q_4$ and $q_2+q_3$, respectively. Thus, the single-pair cost is at
most $L_7$. For a double-pair outcome, the three type-$A$ child sizes sum to
three. Its possible patterns are
 $(3,0,0),(2,1,0),(1,1,1)$,
and their costs are at most
 $f_A(3)=1+2z_A\leq3$, 
 $f_A(2)+f_A(1)=2(1-x_A)+1\leq3$,
 $3f_A(1)=3$, 
respectively.
Hence, a double-pair outcome costs at most three, and either kind of outcome
costs at most $H_7$. The $z_T$-branch, of probability $z_T$, always produces
a single-pair outcome and therefore costs at most $L_7$. Every other branch,
of total probability $1-z_T$, costs at most $H_7$. Consequently,
\[
 f_T(7)\leq z_TL_7+(1-z_T)H_7.
\]

Combining these calculations with \eqref{eq:nc-linear-target}, we impose the
following six base inequalities for each $T\in\{A,B\}$:
\begin{equation}\label{eq:nc-base-system}
\begin{aligned}
 2(1-x_T)&\leq2c+d_T,\\
 1+2z_T&\leq3c+d_T,\\
 2(1-2y_T)&\leq4c+d_T,\\
 (1-2y_T)(1+q_2)+2y_T&\leq5c+d_T,\\
 (1-2y_T)L_6+4y_T&\leq6c+d_T,\\
 z_TL_7+(1-z_T)H_7&\leq7c+d_T.
\end{aligned}
\end{equation}

For the remaining constraints, fix $r\geq8$ and suppose that
\eqref{eq:nc-linear-target} holds for both types at every smaller argument of
size at least two.

\medskip
\noindent\textbf{(C) The constraints on single-pair outcomes.}

A single-pair outcome has two positive type-$B$ child sizes summing to $r-2$.
If the outcome is bad, one child has size one and the other has size $r-3$.
Using $f_B(1)=1$ and the inductive hypothesis for the other child gives
 $f_B(1)+f_B(r-3)
 \leq cr+\bigl(d_B+1-3c\bigr)$.
If the outcome is good, let its child sizes be $r_1,r_2\geq2$, where
$r_1+r_2=r-2$. Applying the inductive hypothesis to both children gives
 $f_B(r_1)+f_B(r_2)
 \leq(cr_1+d_B)+(cr_2+d_B)
 =cr+\bigl(2d_B-2c\bigr)$.
Accordingly, define the bad and good single-pair intercepts by
\begin{equation}\label{eq:nc-single-intercepts}
 S_b=d_B+1-3c,\qquad S_g=2d_B-2c.
\end{equation}

\medskip
\noindent\textbf{(D) The constraints on double-pair outcomes.}

A double-pair outcome has three type-$A$ child sizes summing to $r-4$. A bad
double-pair outcome has no empty child. We distinguish the possible numbers
of singleton children.

If there are two singleton children, the third child has size $r-6\geq2$,
and the cost is at most
 $1+1+f_A(r-6)
 \leq cr+(d_A+2-6c)$.
If there is exactly one singleton child, the other two children have sizes
$r_1,r_2\geq2$ with $r_1+r_2=r-5$. Their total cost is at most
 $1+f_A(r_1)+f_A(r_2) \leq cr+(2d_A+1-5c)$.
If there is no singleton child, all three child sizes $r_1,r_2,r_3$ are at
least two and sum to $r-4$. Thus,
    $f_A(r_1) + f_A(r_2) + f_A(r_3) \le cr+(3d_A-4c)$.
We therefore introduce $D_b$ satisfying
\begin{equation}\label{eq:nc-bad-double-intercepts}
 D_b\geq d_A+2-6c,\qquad
 D_b\geq2d_A+1-5c,\qquad
 D_b\geq3d_A-4c.
\end{equation}

A good double-pair outcome has at least one empty child. If two children are
empty, the remaining child has size $r-4$, and its cost is at most
 $f_A(r-4)\leq cr+(d_A-4c)$.
If exactly one child is empty and one of the remaining children is a
singleton, the last child has size $r-5$, and the cost is at most
 $f_A(1)+f_A(r-5) \leq cr+(d_A+1-5c)$.
Finally, if exactly one child is empty and the other two have sizes
$r_1,r_2\geq2$ with $r_1+r_2=r-4$, then
    $f_A(r_1)+f_A(r_2) \leq cr+(2d_A-4c)$.
We therefore introduce $D_g$ satisfying
\begin{equation}\label{eq:nc-good-double-intercepts}
 D_g\geq d_A-4c,\qquad
 D_g\geq d_A+1-5c,\qquad
 D_g\geq2d_A-4c.
\end{equation}

\medskip
\noindent\textbf{(E) The ordering constraints.}

We impose
\begin{equation}\label{eq:nc-order}
 D_b\geq S_b\geq S_g\geq D_g.
\end{equation}
This ordering also lets
us replace the probability of a favorable branch by its lower bound $y_T$:
the remaining probability can be assigned to the more expensive category
$D_b$, producing a valid upper bound.

\medskip
\noindent\textbf{(F) The recurrence constraints.}

We now translate the structured alternatives in \cref{lem:four-point} into
constraints. In every calculation below, the coefficients of $cr$ sum to
one, so the common term is exactly $cr$.

First, suppose that the four points have a triangular convex hull and the
$z_T$-branch is a good single-pair outcome. Every other branch is a
double-pair outcome, which we upper-bound as bad. The expected cost is at
most
\[
\begin{aligned}
 z_T(cr+S_g)+(1-z_T)(cr+D_b)
 &=cr+z_TS_g+(1-z_T)D_b.
\end{aligned}
\]

Second, suppose that the triangular-hull $z_T$-branch is bad. By
\cref{lem:four-point}, some other branch is a good double-pair outcome. Let
$p\geq y_T$ be the probability of such a branch. Charging all remaining
branches as bad double pairs gives
\[
 z_T(cr+S_b)+p(cr+D_g)+(1-z_T-p)(cr+D_b)
 \leq cr+z_TS_b+y_TD_g+(1-z_T-y_T)D_b,
\]
where the inequality follows from $D_g\leq D_b$ and $p\geq y_T$.

Third, suppose that the four points are in convex position and the
$z_T$-branch is a good single-pair outcome. Let $p\geq y_T$ be the
probability of the unique other single-pair branch. Charging that branch as
bad and the two remaining branches as bad double pairs gives
\[
    z_T(cr+S_g)+p(cr+S_b)+(1-z_T-p)(cr+D_b)
    \leq cr+z_TS_g+y_TS_b+(1-z_T-y_T)D_b,
\]
where we use $S_b\leq D_b$.

Fourth, suppose that the unique non-$z_T$ single-pair branch is good, while
the $z_T$-branch is bad. If the good branch has probability $p\geq y_T$,
then
\[
 z_T(cr+S_b)+p(cr+S_g)+(1-z_T-p)(cr+D_b)
 \leq cr+z_TS_b+y_TS_g+(1-z_T-y_T)D_b,
\]
where we use $S_g\leq D_b$.

Finally, suppose that both single-pair outcomes are bad and some double-pair
branch is good. Let $p_s\geq y_T$ be the probability of the non-$z_T$
single-pair branch and let $p_g\geq y_T$ be the probability of a good
double-pair branch. Then the expected cost is at most
\[
\begin{aligned}
 &z_T(cr+S_b)+p_s(cr+S_b)+p_g(cr+D_g)  +(1-z_T-p_s-p_g)(cr+D_b)\\
 &=cr+(z_T+p_s)S_b+p_gD_g+(1-z_T-p_s-p_g)D_b\\
 &\leq cr+(z_T+y_T)S_b+y_TD_g+(1-z_T-2y_T)D_b.
\end{aligned}
\]
The last inequality follows from $S_b\leq D_b$, $D_g\leq D_b$, and
$p_s,p_g\geq y_T$.

Consequently, for each $T\in\{A,B\}$, we impose the five recurrence
inequalities
\begin{equation}\label{eq:nc-recurrence-system}
\begin{aligned}
 R_1(T)&=z_TS_g+(1-z_T)D_b\leq d_T,\\
 R_2(T)&=z_TS_b+y_TD_g+(1-z_T-y_T)D_b\leq d_T,\\
 R_3(T)&=z_TS_g+y_TS_b+(1-z_T-y_T)D_b\leq d_T,\\
 R_4(T)&=z_TS_b+y_TS_g+(1-z_T-y_T)D_b\leq d_T,\\
 R_5(T)&=(z_T+y_T)S_b+y_TD_g
              +(1-z_T-2y_T)D_b\leq d_T.
\end{aligned}
\end{equation}
All coefficients in \eqref{eq:nc-recurrence-system} are nonnegative by
constraint set~(A). 

A tuple
\[
 (x_A,y_A,z_A,x_B,y_B,z_B,c,d_A,d_B,
  q_2,q_3,q_4,L_6,L_7,H_7,S_b,S_g,D_b,D_g)
\]
is called \emph{feasible} if it satisfies the constraint sets
\textbf{(A)}--\textbf{(F)} above. The optimization problem is to minimize $c$
over all feasible tuples.

\begin{proposition}\label{prop:nc-certificate}
Every feasible tuple satisfies
 $f_T(r)\leq cr+d_T$,
for $T\in\{A,B\}$ and $r\geq2$.
\end{proposition}

\begin{proof}
The calculations in constraint set~\textbf{(B)}, together with the six base
inequalities in \eqref{eq:nc-base-system}, prove the assertion for both types
and every $2\leq r\leq7$.

Now let $r\geq8$ and assume simultaneously that the assertion holds for both
types at every smaller argument of size at least two. By the calculations in
constraint sets~\textbf{(C)} and~\textbf{(D)}, a bad single-pair, good
single-pair, bad double-pair, and good double-pair outcome has expected cost
at most
\[
 cr+S_b,\qquad cr+S_g,\qquad cr+D_b,\qquad cr+D_g,
\]
respectively. The exact cost of every singleton child has already been
included in these expressions.

By \cref{lem:four-point}, every configuration satisfies one of the five
alternatives calculated in constraint set~\textbf{(F)}. For the applicable
index $j$, \eqref{eq:nc-recurrence-system} gives $R_j(T)\leq d_T$.
Hence,
 $f_T(r)\leq cr+R_j(T)\leq cr+d_T$.
\end{proof}

\begin{theorem}\label{thm:nc-randomized}
    The asymmetric randomized algorithm proposed in Section \ref{sec:rand_alg} has an asymptotic competitive ratio $\frac{33517}{50000}=0.67034$
\end{theorem}

\begin{proof}
Choose the probability parameters
\begin{equation}\label{eq:nc-prob-certificate}
\begin{array}{c|ccc}
T&x_T&y_T&z_T\\ \hline
A&\dfrac{193567}{500000}
 &\dfrac{40101}{200000}
 &\dfrac{13241}{62500}\\[2mm]
B&\dfrac12
 &\dfrac{37163}{250000}
 &\dfrac{25337}{125000}
\end{array}
\end{equation}
and choose
\begin{equation}\label{eq:nc-potential-certificate}
 c=\frac{16483}{50000},\qquad
 d_A=\frac{35401}{62500},\qquad
 d_B=\frac{285931}{500000}.
\end{equation}
For the auxiliary variables, take
\[
\begin{gathered}
 q_2=1,\qquad
 q_3=q_4=\frac{87837}{62500},\qquad
 L_6=L_7=\frac{150337}{62500},\qquad H_7=3,\\
 S_b=\frac{291441}{500000},\qquad
 S_g=\frac{121101}{250000},\qquad
 D_b=\frac{73557}{125000},\qquad
 D_g=-\frac{20471}{250000}.
\end{gathered}
\]
One can directly verify that these values satisfy all the constraints in
\textbf{(A)}--\textbf{(F)}. Thus, they form a feasible tuple. Proposition \ref{prop:nc-certificate} then finishes the proof.

\end{proof}

\subsection{A \texorpdfstring{$6/7$}{6/7} upper bound for randomized algorithms}

\begin{theorem}\label{thm:nc-random-upper}
No randomized online algorithm for noncrossing matching has asymptotic
competitive ratio greater than $6/7$.  This holds even when all input points
lie on one circle.
\end{theorem}

\begin{proof}
We apply Yao's principle; the task is to construct an oblivious random input with all points on a circle. Let $h\geq2$ be even. Place $h$ seed points $v_1,\ldots,v_h$ clockwise on a circle and reveal them first in that order. With probability $1/5$, these are all the points, and with probability $4/5$, the $h$ points will be extended randomly as follows. For each open circular arc $A_i$ from
$v_i$ to $v_{i+1}$, independently toss a fair bit $B_i$.  If $B_i=1$, place
one point $w_i$ in $A_i$. Reveal $w_i$ successively.

Fix a deterministic algorithm and suppose it matches $\ell$ segments on the $h$ seed points.  The $\ell$ noncrossing chords divide the disk into $\ell+1$ regions, which we call faces.
Every face contains at least one open arc $A_i$. 
Now, with probability $1/5$, the instance has ended and the algorithm has $h-2\ell$ unmatched points. With probability $4/5$, the instance continues with random points $w_i$. We will show the following

\smallskip

\noindent Claim: at the end of the instance, every face contains in expectation at least $1/2$ unmatched points. 

\smallskip

Assume the Claim. By linearity of expectation, the number of unmatched points is in expectation at least $(\ell+1)/2$. Hence, overall, the expected number of unmatched points is at least
   $\frac15 \cdot (h-2\ell)+\frac45 \cdot \frac{\ell+1}{2}
       =\frac{h+2}{5}$.
Since the expected number of online points is $\frac{1}{5} \cdot h + \frac{4}{5} \cdot \frac{3h}{2} = 7h/5$, and $\OPT$ can match all of them (let us ignore the parity issue which is insignificant to the analysis), the number of unmatched points is at least asymptotically $1/7$ as claimed.  

It remains to prove the Claim. For a face $F$, let $u_F$ be its number of unmatched seed
points and let $J_F$ be its nonempty set of boundary-arc indices.  Hence, the number of points available for further matching in $F$ is 
    $u_F+\sum_{i\in J_F}B_i$.
Observe that this number is even or odd with equal probability. In the case when it is odd, trivially, at least one point cannot be matched. This proves the Claim.
\end{proof}

\section{Open problems}\label{sec:open}

Our work leaves many problems open. We list, in particular, the following:

\begin{itemize}
    \item For online $k$-crossing matching with $k\ge 1$, is there a deterministic algorithm with asymptotic competitive ratio strictly greater than $2/3$?

    \item For online triangle-free-crossing matching, does Greedy have asymptotic competitive ratio at least $2/3$? We believe it does. This is equivalent to the following.

\begin{conjecture}  \label{conj:main_conj}
    For every $\ell\in \N$, $\beta_\triangle(\ell) \le \ell$.
\end{conjecture}

    Perhaps the more general version of $\beta_{\mathsf P}(\ell)$ is also worth studying, such as when ${\mathsf P}$ is the $K_t$-free property, i.e., no $t$ line segments pairwise crossing. 

    \item Prove or disprove Conjecture \ref{conj:vtx_conj}.

    \item Is there a deeper connection between $\beta_\triangle$ and $R(3,t)$? An interesting feature of our constructions  in Section \ref{sec:constructions} is that larger constructions share simpler admissible gadgets, essentially only those listed in Figure \ref{fig:gadgets}, of size 3, 4, 5, 8, respectively. Incidentally, the gaps between Ramsey numbers also satisfy $R(3,t) - R(3,t-1) \in \{3, 4,5,8\}$ for $3 \le t \le 10$ (either $R(3,10) = 40$ or $41$). Is this a coincidence? 
    
    This also points to the question of better understanding of the sequence of Ramsey gaps $g(t) := R(3,t) - R(3,t-1)$. By Kim \cite{kim1995ramsey}, $R(3,t) = \Theta(t^2/\log t)$, hence, the Ramsey gap has average of order $\Theta(t/\log t)$. Is there any pattern for the sequence $g(t)$, and does it connect to a possible sequence of basic building blocks (e.g., those admissible simple gadgets) in our geometric problem?  One may also ask the Ramsey gap problem for  general Ramsey numbers.

    \item For online noncrossing matching, what is the optimal competitive ratio of randomized algorithms? 
\end{itemize}

\section*{Declaration of AI Usage}
No AI tool is used for the proposal of the model and analysis of the results for matching with limited crosses, i.e., Section \ref{sec:k-crossing} and Section \ref{sec:triangle-free}.

The AI Tool ChatGPT is used to develop the results in Section \ref{sec:noncrossing}. Initially, AI's solution and presentation for Section \ref{sec:rand_alg} are filled with optimization with specific numerical numbers.    The authors have changed the presentation to its current form: a structure lemma followed by the optimization problem with a clear separation of different groups of constraints, and replaced specific numerical values by variables. We hope this facilitates the understanding. The authors have verified the proofs and take full responsibility for the contents of this article.

\medskip
{\noindent\bf Acknowledgment.} The authors thank the reviewers' advice that helps improve the presentation.

\bibliographystyle{alpha}
\bibliography{mybib}

\end{document}